\documentclass[runningheads,a4paper]{llncs}

\usepackage{amssymb}
\usepackage{graphicx}
\usepackage{url}
\urldef{\mails}\path|meizeh@cs.technion.ac.il| 
\usepackage{ntheorem}
\theoremseparator{.}
\newtheorem{obs}{Observation}
\newcommand{\negA}{\vspace{-0.05in}}

\newcommand{\negC}{\vspace{-0.18in}}

\newcommand{\mysection}[1]{\negC\section{#1}\negA}

\newcommand{\myparagraph}[1]{\par\smallskip\par\noindent{\bf{}#1:~}}

\usepackage{algorithm}
\usepackage{algorithmic}
\newcommand{\alg}[1]{\mbox{\sf #1}}  

\usepackage[left=1in,right=1in,top=1in,bottom=1in]{geometry}

\begin{document}

\mainmatter

\title{Deterministic Parameterized Algorithms for Matching and Packing Problems}

\author{Meirav Zehavi}

\institute{Department of Computer Science, Technion - Israel Institute of Technology, Haifa 32000, Israel\\
\mails}

\maketitle

\begin{abstract}
We present three deterministic parameterized algorithms for well-studied packing and matching problems, namely, {\em Weighted $q$-Dimensional $p$-Matching (($q,p$)-WDM)} and {\em Weighted $q$-Set $p$-Packing (($q,p$)-WSP)}. More specifically, we present an $O^*(2.85043^{(q-1)p})$ time deterministic algorithm for $(q,p)$-WDM, an $O^*(8.04143^p)$ time deterministic algorithm for the unweighted version of $(3,p)$-WDM, and an $O^*((0.56201\cdot2.85043^{q})^p)$ time deterministic algorithm for $(q,p)$-WSP. Our algorithms significantly improve the previously best known $O^*$ running times in solving $(q,p)$-WDM and $(q,p)$-WSP, and the previously best known deterministic $O^*$ running times in solving the unweighted versions of these problems. Moreover, we present kernels of size $O(e^qq(p-1)^q)$ for ($q,p$)-WDM and ($q,p$)-WSP, improving the previously best known kernels of size $O(q!q(p-1)^q)$ for these problems.
\end{abstract}

\mysection{Introduction}

We consider the following well-studied matching and packing problems.

\smallskip

{\noindent \underline{Weighted $q$-Dimensional $p$-Matching (($q,p$)-WDM)}}
\begin{itemize}
\item Input: Pairwise disjoint universes $U_1,\ldots,U_q$, a set ${\cal S}\subseteq U_1\times\ldots\times U_q$, a weight function $w: {\cal S}\rightarrow\mathbb{R}$, and a parameter $p$.
\item Output: A subset ${\cal S}'\subseteq {\cal S}$ of $p$ disjoint tuples, which maximizes $\sum_{S\in {\cal S}'}w(S)$.
\end{itemize}

{\noindent \underline{Weighted $q$-Set $p$-Packing (($q,p$)-WSP)}}
\begin{itemize}
\item Input: A universe $U$, a set ${\cal S}$ of subsets of size $q$ of $U$, a weight function $w: {\cal S}\rightarrow\mathbb{R}$, and a parameter~$p$.
\item Output: A subset ${\cal S}'\subseteq {\cal S}$ of $p$ disjoint sets, which maximizes $\sum_{S\in {\cal S}'}w(S)$.
\end{itemize}
The {\em $q$-Dimensional $p$-Matching (($q,p$)-DM)} problem is the special case of ($q,p$)-WDM in which all of the tuples in ${\cal S}$ have the same weight. Similarly, the {\em $q$-Set $p$-Packing (($q,p$)-SP)} problem is the special case of ($q,p$)-WSP in which all of the tuples in ${\cal S}$ have the same weight. Note that ($q,p$)-WDM is a special case of ($q,p$)-WSP.

As noted by Chen et al. \cite{impdetmatpac}, matching and packing problems form an important class of NP-hard problems. In particular, the six "basic" NP-complete problems include 3-Dimensional Matching \cite{subgraphisoneg}.

A {\em parameterized algorithm} solves an NP-hard problem by confining the combinatorial explosion to a parameter $k$. More precisely, a problem is {\em fixed-parameter tractable (FPT)} with respect to a parameter $k$ if an instance of size $n$ can be solved in time $O^*(f(k))$ for some function $f(k)$ \cite{fixedpar}.\footnote{$O^*$ hides factors polynomial in the input size.} A {\em kernelization algorithm} for a problem $P$ is a polynomial-time algorithm that, given an instance $x$ of $P$ and a parameter $k$, returns an instance $x'$ of $P$ whose size is bounded by some function $f(k)$, such that there is a solution to $x$ iff there is a solution to $x'$. We then say that $P$ has a kernel of size $f(k)$.

In this paper we present three deterministic parameterized algorithms and deterministic kernelization algorithms for ($q,p$)-WDM and ($q,p$)-WSP, where the parameter is $(p+q)$.

\begin{table}[center]
\centering
\begin{tabular}{|l|c|c|c|}
	\hline
	Reference 		              & Randomized$\setminus$Deterministic & Variation   & Running Time                          \\ \hline \hline
	Chen et al. \cite{chenalgorithmica2004}	    		& D              & $(q,p)$-SP  & $O^*((qp)^{O(qp)})$     							 \\ \hline
	Downey et al. \cite{fellowsbook}	    		      & D              & $(q,p)$-WSP & $O^*((qp)^{O(qp)})$     							 \\ \hline
	Fellows et al. \cite{fellowsalgorithmica2008}	  & D              & $(q,p)$-WSP & $O^*(2^{O(qp)})$     							   \\ \hline					
	Koutis \cite{koutis2005}												& D              & $(q,p)$-SP  & $O^*(2^{O(qp)})$    							     \\
																									& R              & $(q,p)$-SP  & $O^*(10.874^{qp})$    							   \\ \hline
	Chen et al. \cite{divandcol}	     	            & D              & $(q,p)$-WSP & $O^*(4^{qp+o(qp)})$       						 \\
																     						  & R              & $(q,p)$-WSP & $O^*(4^{(q-1)p+o(qp)})$               \\ \hline	
	Chen et al. \cite{impdetmatpac}	     						& D              & $(q,p)$-WSP & $O^*(4^{(q-0.5)p+o(qp)})$             \\ 
																	     						& D              & $(q,p)$-WDM & $O^*(4^{(q-1)p+o(qp)})$               \\ \hline
  Koutis \cite{multilineardetection}							& R              & $(q,p)$-SP  & $O^*(2^{qp})$ 		    	               \\ \hline	
  Koutis et al. \cite{appmultilinear}							& R              & $(q,p)$-DM  & $O^*(2^{(q-1)p})$ 			               \\ \hline	
	Bj$\ddot{\mathrm{o}}$rklund et al. \cite{bjo10}	& R              & $(q,p)$-DM  & $O^*(2^{(q-2)p})$ 			               \\ \hline				     						{\bf This paper}											     		  & {\bf D}        & $\bf(q,p)${\bf-WSP} & $\bf O^*((0.563\cdot 2.851^q)^{p})$ \\
																	     						& {\bf D}        & $\bf(q,p)${\bf-WDM} & $\bf O^*(2.851^{(q-1)p})$           \\ \hline							 	
\end{tabular}\smallskip
\caption{Known parameterized algorithms for ($q,p$)-WDM and ($q,p$)-WSP.}
\label{tab:knownresults}
\end{table}

\begin{table}[center]
\centering
\begin{tabular}{|l|c|c|c|}
	\hline
	Reference 		              & Randomized$\setminus$Deterministic & Variation   & Running Time              \\ \hline \hline
	Chen et al. \cite{chenalgorithmica2004}	    		& D              & $(3,p)$-SP  & $O^*(p^{O(p)})$           \\ \hline	
	Downey et al. \cite{fellowsbook}	    		      & D              & $(3,p)$-WSP & $O^*(p^{O(p)})$     		   \\ \hline
	Fellows et al. \cite{fellowsalgorithmica2008}	  & D              & $(3,p)$-WSP & $O^*(2^{O(p)})$     		   \\ \hline
	Liu et al. \cite{liutamc2007}	    				 	    & D              & $(3,p)$-WSP & $O^*(2,0097.152^{p})$     \\ \hline
	Koutis \cite{koutis2005}												& D              & $(3,p)$-SP  & $O^*(2^{O(qp)})$    			 \\
																									& R              & $(3,p)$-SP  & $O^*(1,285.475^{p})$      \\ \hline	
	Wang et al. \cite{wangcocoon2008}	    				 	& D              & $(3,p)$-WSP & $O^*(432.082^{p})$        \\ \hline
	Chen et al. \cite{predivandcol}	     						& D              & $(3,p)$-DM  & $O^*(21.907^{p})$         \\ \hline	
	Liu et al. \cite{chenipec}	     								& D              & $(3,p)$-SP  & $O^*(97.973^p)$       		 \\
																									& D              & $(3,p)$-DM  & $O^*(21.254^p)$       		 \\
																									& R              & $(3,p)$-DM  & $O^*(12.488^p)$        	 \\ \hline
	Chen et al. \cite{divandcol}	     							& D              & $(3,p)$-WSP & $O^*(64^{p+o(p)})$        \\ 
																     						  & R              & $(3,p)$-WSP & $O^*(16^{p+o(p)})$        \\ \hline	
	Wang et al. \cite{wangtamc2008}	    				  	& D              & $(3,p)$-SP  & $O^*(43.615^{p})$         \\ \hline
	Chen et al. \cite{impdetmatpac}	     						& D              & $(3,p)$-WSP & $O^*(32^{p+o(p)})$        \\ 	
			     														 						& D              & $(3,p)$-WDM & $O^*(16^{p+o(p)})$        \\ \hline
  Koutis \cite{multilineardetection}			        & R              & $(3,p)$-SP  & $O^*(8^{p})$              \\ \hline	
  Koutis et al. \cite{appmultilinear}							& R              & $(3,p)$-DM  & $O^*(4^{p})$              \\ \hline	
	Bj$\ddot{\mathrm{o}}$rklund et al. \cite{bjo10}	& R              & $(3,p)$-SP  & $O^*(3.344^{p})$      		 \\ 				     																																			        & R              & $(3,p)$-DM  & $O^*(2^{p})$      			   \\ \hline	
	{\bf This paper}											 		  		& {\bf D}        & $\bf(3,p)${\bf-WSP} & $\bf O^*(12.155^{p})$       				 \\ 	
			     														 						& {\bf D}        & $\bf(3,p)${\bf-WDM} & $\bf O^*(8.125^{p})$                \\ 	
			     														 						& {\bf D}        & $\bf(3,p)${\bf-DM}  & $\bf O^*(8.042^{p})$                \\ \hline							 	
\end{tabular}\smallskip
\caption{Known parameterized algorithms for ($3,p$)-WDM and ($3,p$)-WSP.}
\label{tab:knownresults3}
\end{table}

\myparagraph{Prior Work and Our Contribution}A lot of attention has been paid to ($q,p$)-WDM and ($q,p$)-WSP. Tables~\ref{tab:knownresults} and \ref{tab:knownresults3} present a summary of parameterized algorithms for these problems. In particular, Chen et al. \cite{impdetmatpac} gave a deterministic algorithm for $(q,p)$-WDM that runs in time $O^*(4^{(q-1)p+o(qp)})$. This algorithm has the previously best known $O^*$ running time for $(q,p)$-WDM (for any $q$), and the previously best known deterministic $O^*$ running time for $(q,p)$-DM (for any $q$). Our first result is a deterministic algorithm for $(q,p)$-WDM that runs in time $O^*(2.851^{(q-1)p})$. We thus achieve a significant improvement over the previously best known $O^*$ running time for $(q,p)$-WDM (for any $q$), and the previously best known deterministic $O^*$ running time for $(q,p)$-DM (for any $q$). Our second result is a deterministic algorithm for $(3,p)$-DM, which further reduces the $O^*$ running time of our first algorithm, when applied to $(3,p)$-DM, from $O^*(8.125^p)$ to $O^*(8.042^p)$.

Chen et al. \cite{divandcol} gave a randomized algorithm for $(q,p)$-WSP that runs in time $O^*(4^{(q-0.1)p+o(qp)})$, and Chen et al. \cite{impdetmatpac} gave a deterministic algorithm for $(q,p)$-WSP that runs in time $O^*(4^{(q-0.5)p+o(qp)})$. These algorithms have the previously best known $O^*$ running time for $(q,p)$-WSP (for any $q$), and the previously best known deterministic $O^*$ running time for $(q,p)$-SP (for any $q$). Our third result is a deterministic algorithm for $(q,p)$-WSP that runs in time $O^*((0.563\cdot 2.851^q)^{p})$, where for the special case of $(3,p)$-WSP, it runs in time $O^*(12.155^{p})$. We thus achieve a significant improvement over the previously best known $O^*$ running time for $(q,p)$-WSP (for any $q$), and the previously best known deterministic $O^*$ running time for $(q,p)$-SP (for any $q$).

Assuming that $q=O(1)$, Chen et al. \cite{impdetmatpac} gave kernels of size $O(q^qqp^q)$ for $(q,p)$-WDM and $(q,p)$-WSP. Fellows et al. \cite{fellowsalgorithmica2008} gave kernels of size $O(q!q(p-1)^q)$ for $(q,p)$-DM and $(q,p)$-SP, which can be extended to kernels of the same size for $(q,p)$-WDM and $(q,p)$-WSP. Dell et al. \cite{holmar2012} proved that $(q,p)$-DM is unlikely to admit a kernel of size $O(f(q)p^{q-\epsilon})$ for any function $f(q)$ and $\epsilon>0$ (improving upon a result by Hermelin et al. \cite{herwu2012}). Our fourth result presents kernels of size $O(e^qq(p-1)^q)$ for $(q,p)$-WDM and $(q,p)$-WSP.

\myparagraph{Organization}Section \ref{section:repset} gives some background about representative sets and two related results by Fomin et al. \cite{representative}. Sections \ref{section:wdm}, \ref{section:dm} and \ref{section:wsp} present deterministic algorithms for $(q,p)$-WDM, $(3,p)$-DM and $(q,p)$-WSP, respectively. Finally, Section \ref{section:kernel} gives kernels for $(q,p)$-WDM and $(q,p)$-WSP, and uses them to improve the running times of the algorithms presented in the previous three sections.

\mysection{Representative Sets}\label{section:repset}

Recently, Fomin et al. \cite{representative} presented two new efficient computations of representative sets, which they then used to design improved deterministic parameterized algorithms for "graph connectivity" problems such as $k$-Path (i.e., finding a path of length at least $k$ in a given graph). Our algorithms rely on these results, which we present in this section.

\begin{definition}
Let $U$ be a universe, $s,r\in\mathbb{Z}$, and $\cal A$ be a set of triples $(X,{\cal S}',W)$ s.t. $X\subseteq U$, $|X|=s$ and $W\in\mathbb{R}$.\\
We say that a subset $\widehat{\cal A}\subseteq {\cal A}$ {\em (max) $r$-represents $\cal A$} if for every $Y\subseteq U$ s.t. $|Y|\leq r$ the following holds: if there is $(X,{\cal S}',W)\in {\cal A}$ s.t. $X\cap Y=\emptyset$, then there is $(X^*,{\cal S}^*,W^*)\in\widehat{\cal A}$ s.t. $X^*\cap Y=\emptyset$ and $W^*\geq W$.
\end{definition}
By Section 4.2 in \cite{representative}, we have a deterministic algorithm, that we call \alg{R-Alg}($U,s,r,{\cal A}$), whose input, output and running time are as follows.

\begin{itemize}
\item Input: A universe $U$, $s,r\in\mathbb{Z}$, and a set $\cal A$ of triples $(X,{\cal S}',W)$ s.t. $X\subseteq U$, $|X|=s$ and $W\in\mathbb{R}$.
\item Output: A subset $\widehat{\cal A}\subseteq {\cal A}$ s.t. $|\widehat{{\cal A}}|\leq {s+r \choose s}2^{o(s+r)} \log |U|$, which $r$-represents $\cal A$.
\item Running time: $O(|{\cal A}|(\frac{s+r}{r})^r\log |U|)$.
\end{itemize}
By Section 4.1 in \cite{representative}, we have a deterministic algorithm, that we call \alg{K-Alg}($U,s,r,{\cal A}$), whose input, output and running time are as follows.

\begin{itemize}
\item Input: A universe $U$, $s,r\in\mathbb{Z}$, and a set $\cal A$ of triples $(X,{\cal S}',W)$ s.t. $X\subseteq U$, $|X|=s$ and $W\in\mathbb{R}$.
\item Output: A subset $\widehat{\cal A}\subseteq {\cal A}$ s.t. $|\widehat{{\cal A}}|\leq {s+r \choose s}$, which $r$-represents $\cal A$.
\item Running time: $O(|{\cal A}|{s+r \choose s}^{\tilde{w}-1}$$\log(s!|U|^{s^2}))$, where $\tilde{w}$$<$$2.373$ is the matrix multiplication exponent~\cite{vir2012}.
\end{itemize}
We also need the following observation from \cite{representative}.

\begin{obs}\label{obs:transitive}
Let $U$ be a universe, $s,r\in\mathbb{Z}$, and ${\cal A},\widehat{\cal A}$ and $\widehat{\cal A}'$ be sets of triples $(X,{\cal S}',W)$ s.t. $X\subseteq U$, $|X|=s$ and $W\in\mathbb{R}$. If $\widehat{\cal A}'$ $r$-represents $\widehat{\cal A}$ and $\widehat{\cal A}$ $r$-represents $\cal A$, then $\widehat{\cal A}'$ $r$-represents $\cal A$.
\end{obs}

\mysection{An Algorithm for ($q,p$)-WDM}\label{section:wdm}

Let $<$ be an order on $U_1$. Roughly speaking, the idea of the algorithm is to iterate over $U_1$ in an ascending order, such that when we reach an element $u\in U_1$, we have already computed representative sets of sets of "partial solutions" that include only tuples whose first elements are smaller than $u$. Then, we try to extend the "partial solutions" by adding tuples whose first element is $u$ and computing new representative sets accordingly. Note that the elements in $U_1$ that appear in the "partial solutions" do not appear in any tuple whose first element is at least $u$, and that any tuple whose first element is at least $u$ does not contain elements in $U_1$ that appear in the "partial solutions". This allows us to use "better" representative sets, which improves the running time of the algorithm.

We next give the notation used in this section. We then describe the algorithm and give its pseudocode. Finally, we prove its correctness and running time.

\myparagraph{Notation}Denote $U=U_1\cup\ldots\cup U_q$. Let $u_s$ (resp. $u_g$) be the smallest (resp. greatest) element in $U_1$. Given $u\in U_1$, denote ${\cal S}_u=\{S\in {\cal S}: S$ includes $u\}$. Given a tuple $S$, let $\mathrm{set}(S)$ be the set of elements in $S$, excluding its first element. Given a set of tuples ${\cal S}'$, denote $\mathrm{tri}({\cal S}')=(\bigcup_{S\in{\cal S}'}\mathrm{set}(S),{\cal S}',\sum_{S\in{\cal S}'}w(S))$. Given a set of sets of tuples ${\bf S}$, denote $\mathrm{tri}({\bf S})=\{\mathrm{tri}({\cal S}'): {\cal S}'\in {\bf S}\}$. Given $S\in{\cal S}$ and $1\leq j\leq q$, let $S_j$ denote the tuple including the first $j$ elements in $S$, and define $w(S_j)=w(S)$.

Given $u\in U_1$ and $1\leq i\leq p$, let $SOL_{u,i}$ be the set of all sets of $i$ disjoint tuples in $\cal S$ whose first elements are at most $u$ (i.e., $SOL_{u,i}= \{{\cal S}'\subseteq\bigcup_{u'\in U_1\ \mathrm{s.t.}\ u'\leq u}{\cal S}_{u'}: |{\cal S}'|=i,$ the tuples in ${\cal S}'$ are disjoint$\}$).  Note that for all $(X,{\cal S}',W)\in\mathrm{tri}(SOL_{u,i})$, we have that $|X|=(q-1)i$. Given also $S\in {\cal S}_u$ and $1\leq j\leq q$, let $SOL_{u,i,S,j}$ be the set of all sets of disjoint tuples that include $S_j$ and $i-1$ tuples in $\cal S$ whose first elements are smaller than $u$ (i.e., $SOL_{u,i,S,j}= \{{\cal S}'\subseteq\{S_j\}\cup(\bigcup_{u'\in U_1\ \mathrm{s.t.}\ u'<u}{\cal S}_u'): S_j\in {\cal S}', |{\cal S}'|=i,$ the tuples in ${\cal S}'$ are disjoint$\}$). Note that for all $(X,{\cal S}',W)\in\mathrm{tri}(SOL_{u,i,S,j})$, we have that $|X|=(q-1)(i-1)+j-1$.

\myparagraph{The Algorithm}We now describe our algorithm for ($q,p$)-WDM, that we call \alg{WDM-Alg} (see the pseudocode below). The algorithm starts by introducing a matrix M, where each cell M$[u,i]$ will hold a subset of $SOL_{u,i}$.

\alg{WDM-Alg} iterates over $U_1$ in an ascending order. In each iteration, corresponding to some $u\in U_1$, it computes any cell of the form M$[u,i]$ s.t. $1\leq i\leq p$ by using M$[u',i]$ and M$[u',i-1]$ (where $u'$ is the element preceding $u$ in $U_1$). In other words, for any $1\leq i\leq p$, it computes a subset of $SOL_{u,i}$ by using subsets of $SOL_{u',i-1}$ and $SOL_{u',i}$. If there is a solution, then by using representative sets, \alg{WDM-Alg} guarantees that each cell M$[u,i]$ will hold "enough" sets from $SOL_{u,i}$, such that when the computation of M is finished, M$[u_g,p]$ will hold some ${\cal S}'\in SOL_{u_g,p}$ that maximizes $\sum_{S\in{\cal S}'}w(S)$ (clearly, such a set ${\cal S}'$ is a solution). Moreover, by using representative sets, \alg{WDM-Alg} guarantees that each cell M$[u,i]$ will not hold "too many" sets from $SOL_{u,i}$, since then we will not get an improved running time.

We now describe an iteration, corresponding to some $u\in U_1$, in more detail. By using \alg{R-Alg}, \alg{WDM-Alg} first computes a set that $(q-1)(p-1)$-represents tri($SOL_{u,1}$) (in Step \ref{stepwdm:mu1a}), and assigns its corresponding set of sets of tuples to M$[u,1]$ (in Step \ref{stepwdm:mu1b}). If $u=u_s$, then $SOL_{u,i}$ is empty for all $2\leq i\leq p$, and \alg{WDM-Alg} skips the rest of the iteration accordingly (thus M$[u,i]$ stays empty, as it is initialized, for all $2\leq i\leq p$). Next assume that $u>u_s$, and consider an iteration of the internal loop, corresponding to some $2\leq i\leq p$. First, in Step \ref{stepwdm:step1}, \alg{WDM-Alg} computes a set that $(q-1)(p-i)$-represents tri($SOL_{u,i}$) by using the sets in M$[u',i]$ and adding tuples in ${\cal S}_u$ to sets in M$[u',i-1]$. In particular, for any tuple $S\in {\cal S}_u$, \alg{WDM-Alg} calls \alg{WDM-Add}, which adds $S$ to sets of tuples disjoint from $S$ in M$[u',i-1]$. \alg{WDM-Add} iterates over the elements in $S$, adding them one by one (excluding the first element, which it does not add) to sets in M$[u',i-1]$. After adding each element, \alg{WDM-Add} uses \alg{R-Alg} to compute a representative set of the result.\footnote{This approach results in a running time better than that achieved by adding all the elements of the tuple "at once" and only then using \alg{R-Alg}.} Then, in Step \ref{stepwdm:step2}, \alg{WDM-Alg} uses \alg{R-Alg} to compute a representative set of the representative set it has just computed in Step \ref{stepwdm:step1} in order to reduce its size. Finally, in Step \ref{stepwdm:step3}, \alg{WDM-Alg} assigns the corresponding set of sets of tuples to M$[u,i]$.

\renewcommand{\thealgorithm}{1}
\begin{algorithm}[!ht]
\caption{\alg{WDM-Alg}($U_1,\ldots,U_q,{\cal S},w,p$)}
\begin{algorithmic}[1]\label{alg:wdm}
\STATE let M be a matrix that has a cell $[u,i]$ for all $u\in U_1$ and $1\leq i\leq p$, which is initialized to $\{\}$.
\FORALL{$u\in U_1$ {\bf ascending}}
	\STATE\label{stepwdm:mu1a} $\widehat{\cal A}\Leftarrow$ \alg{R-Alg}($U,q-1,(q-1)(p-1),$tri($\{\{S\}: S\in\bigcup_{u'\in U_1\ \mathrm{s.t.}\ u'\leq u}{\cal S}_{u'}\}$)).
	\STATE\label{stepwdm:mu1b} M$[u,1]\Leftarrow\{{\cal S'}: \exists X,W$ s.t. $(X,{\cal S'},W)\in\widehat{\cal A}\}$.
	\STATE {\bf if} $u=u_s$ {\bf then} skip the iteration. {\bf else} let $u'$ be the element preceding $u$ in $U_1$.
	\FOR{$i=2,\ldots,p$}\label{stepwdm:loop}
		\STATE\label{stepwdm:step1} ${\cal A}\Leftarrow \mathrm{tri}($M$[u',i]) \cup (\bigcup_{S\in {\cal S}_u}$\alg{WDM-Add}($i,S,\mathrm{M}[u',i-1]))$.
		\STATE\label{stepwdm:step2} $\widehat{{\cal A}}\Leftarrow$ \alg{R-Alg}($U,(q-1)i,(q-1)(p-i),{\cal A}$).
		\STATE\label{stepwdm:step3} M$[u,i]\Leftarrow\{{\cal S'}: \exists X,W$ s.t. $(X,{\cal S'},W)\in\widehat{\cal A}\}$.
	\ENDFOR
\ENDFOR
\STATE {\bf if} M$[u_g,p]=\emptyset$ {\bf then reject}. {\bf else return} ${\cal S}'\in\mathrm{M}[u_g,p]$ that maximizes $\sum_{S\in{\cal S}'}w(S)$.
\end{algorithmic}
\end{algorithm}

\begin{algorithm}[!ht]
\caption{\alg{WDM-Add}($i,S,{\bf S}$)}
\begin{algorithmic}[1]\label{alg:wdmadd}
\STATE\label{stepwdmadd:b1} $\widehat{\cal B}_1\Leftarrow \{(X,{\cal S}'\cup\{S_1\}, W+w(S)): (X,{\cal S}',W)\in \mathrm{tri}({\bf S})\}$.
\FOR{$j=2,\ldots,q$}
	\STATE\label{stepwdmadd:bj} ${\cal B}_j\Leftarrow \{(X\cup\{u_j\},({\cal S}'\setminus\{S_{j-1}\})\cup\{S_j\}, W): (X,{\cal S}',W)\in \widehat{\cal B}_{j-1}, u_j$ is the $j^\mathrm{st}$ element in $S, u_j\notin X\}$.
	\STATE\label{stepwdmadd:hatbj} $\widehat{\cal B}_j\Leftarrow$ \alg{R-Alg}($U,(q-1)(i-1)+(j-1),(q-1)(p-i)+(q-j),{\cal B}_j$).
\ENDFOR
\STATE {\bf return} $\widehat{\cal B}_q$.
\end{algorithmic}
\end{algorithm}

\myparagraph{Correctness and Running Time}We start by proving the following lemma regarding \alg{WDM-Add}.

\begin{lemma}\label{lemma:wdmadd}
Given $2\leq i\leq p$, $S\in S_u$ for some $u\in U_1$, and ${\bf S}$ s.t. $\mathrm{tri}({\bf S})$ $(q-1)(p-(i-1))$-represents $\mathrm{tri}(SOL_{u',i-1})$ where $u'$ is the element preceding $u$ in $U_1$, \alg{WDM-Add} returns a set that $(q-1)(p-i)$-represents $\mathrm{tri}(SOL_{u,i,S,q})$.
\end{lemma}

\begin{proof}
By using induction on $j$, we prove that for all $1\leq j\leq q$, $\widehat{\cal B}_j$ $((q-1)(p-i)+(q-j))$-represents tri($SOL_{u,i,S,j}$). By Step \ref{stepwdmadd:b1}, since tri$({\bf S})$ $(q-1)(p-(i-1))$-represents tri($SOL_{u',i-1}$), we have that $\widehat{\cal B}_1$ $((q-1)(p-i)+(q-1))$-represents tri($SOL_{u,i,S,1}$).

Next consider some $2\leq j\leq q$, and assume that the claim holds for all $1\leq j'<j$. By the definition of \alg{R-Alg}, Observation \ref{obs:transitive} and Step \ref{stepwdmadd:hatbj}, it is enough to prove that ${\cal B}_j$ $((q-1)(p-i)+(q-j))$-represents tri$(SOL_{u,i,S,j})$.

By the induction hypothesis and Step \ref{stepwdmadd:bj}, we get that ${\cal B}_j\subseteq \mathrm{tri}(SOL_{u,i,S,j})$. Assume that there are $(X,{\cal S}',W)\in\mathrm{tri}(SOL_{u,i,S,j})$ and $Y\subseteq U\setminus X$ s.t. $|Y|\leq ((q-1)(p-i)+(q-j))$, since otherwise the claim clearly holds. Let $u_j$ be the $j^\mathrm{st}$ element in $S$. Note that $(X\setminus\{u_j\},({\cal S}'\setminus\{S_j\})\cup\{S_{j-1}\},W)\in \mathrm{tri}(SOL_{u,i,S,j-1})$. Thus, by the induction hypothesis, there is $(X^*,{\cal S}^*,W^*)\in\widehat{{\cal B}}_{j-1}$ s.t. $X^*\cap (Y\cup\{u_j\})=\emptyset$ and $W^*\geq W$. We get that $(X^*\cup\{u_j\},({\cal S}^*\setminus\{S_{j-1}\})\cup\{S_j\},W^*)\in{\cal B}_j$. Since $(X^*\cup\{u_j\})\cap Y=\emptyset$ and $W^*\geq W$, we get that the claim holds.\qed
\end{proof}

\begin{theorem}\label{theorem:wdm}
\alg{WDM-Alg} solves ($q,p$)-WDM in $O(2.85043^{(q-1)p}|{\cal S}||U|\log^2|U|)$ deterministic time. In particular, it solves ($3,p$)-WDM in $O^*(8.12492^p)$ deterministic time.
\end{theorem}

\begin{proof}
The following lemma clearly implies the correctness of the algorithm.

\begin{lemma}\label{lemma:wdm}
For all $u\in U_1$ and $1\leq i\leq p$, $\mathrm{tri}(\mathrm{M}[u,i])$ $(q-1)(p-i)$-represents $\mathrm{tri}(SOL_{u,i})$.
\end{lemma}

\begin{proof}
We prove the lemma by using induction on the order of the computation of M. For all $u\in U_1$, $SOL_{u,1} =\{\{S\}: S\in\bigcup_{u'\in U_1\ \mathrm{s.t.}\ u'\leq u}{\cal S}_{u'}\}$; and thus, by the definition of \alg{R-Alg} and Steps \ref{stepwdm:mu1a} and \ref{stepwdm:mu1b}, tri(M$[u,1]$) $(q-1)(p-1)$-represents tri$(SOL_{u,1})$. For all $2\leq i\leq p$, $SOL_{u_s,i} =\{\}$; and thus, by the initialization of M, tri(M$[u_s,i]$) $(q-1)(p-i)$-represents tri$(SOL_{u_s,i})$.

Next consider an iteration of Step \ref{stepwdm:loop} that corresponds to some $u\in U_1\setminus\{u_s\}$ and $2\leq i\leq p$, and assume that the lemma holds for the element $u'$ preceding $u$ in $U_1$ and all $1\leq i'\leq i$. By the definition of \alg{R-Alg}, Observation \ref{obs:transitive} and Steps \ref{stepwdm:step2} and \ref{stepwdm:step3}, it is enough to prove that $\cal A$ $(q-1)(p-i)$-represents tri$(SOL_{u,i})$.

By the induction hypothesis, Step \ref{stepwdm:step1} and Lemma \ref{lemma:wdmadd}, we have that ${\cal A}\subseteq \mathrm{tri}(SOL_{u,i})$. Assume that there are $(X,{\cal S}',W)\in\mathrm{tri}(SOL_{u,i})$ and $Y\subseteq U\setminus X$ s.t. $|Y|\leq (q-1)(p-i)$, since otherwise the lemma clearly holds. We have two possible cases as follows.
\begin{enumerate}
\item ${\cal S}'\cap {\cal S}_u=\emptyset$. Note that ${\cal S}'\in SOL_{u',i}$. Thus, by the induction hypothesis, there is $(X^*,{\cal S}^*,W^*)\in\mathrm{tri(M}[u',i])$ s.t. $X^*\cap Y=\emptyset$ and $W^*\geq W$; and therefore $(X^*,{\cal S}^*,W^*)\in{\cal A}$.
\item ${\cal S}'\cap {\cal S}_u=\{S\}$ for some $S$. Note that ${\cal S}'\in SOL_{u,i,S,q}$. Thus, by the induction hypothesis and Lemma \ref{lemma:wdmadd}, \alg{WDM-Add}$(i,S,\mathrm{M}[u',i-1])$ returns a set that includes a triple $(X^*,{\cal S}^*,W^*)$ s.t. $X^*\cap Y=\emptyset$ and $W^*\geq W$; and therefore $(X^*,{\cal S}^*,W^*)\in{\cal A}$.
\end{enumerate}
We get that there is $(X^*,{\cal S}^*,W^*)\in{\cal A}$ s.t. $X^*\cap Y=\emptyset$ and $W^*\geq W$.\qed
\end{proof}
By the definition of \alg{R-Alg} and the pseudocode, the algorithm runs in time

\[O(\sum_{u\in U_1}\sum_{i=1}^{p}\sum_{j=1}^q[{(q-1)p \choose (q-1)(i-1)+j-1}2^{o(qp)}|{\cal S}|(\frac{(q-1)p}{(q-1)(p-i)+q-j})^{(q-1)(p-i)+q-j}\log^2 |U|)] =\]

\[O(2^{o(qp)}|{\cal S}||U|\log^2|U|\cdot\mathrm{max}_{t=0}^{(q-1)p}\left\{{(q-1)p \choose t}(\frac{(q-1)p}{(q-1)p-t})^{(q-1)p-t}\right\})\]
The maximum is achieved at $i=\alpha (q-1)p$, where $\alpha = 1 + \frac{1-\sqrt{1+4e}}{2e}$. Thus, the running time of the algorithm is $O(2.85043^{(q-1)p}|{\cal S}||U|\log^2|U|)$.\qed
\end{proof}

\mysection{An Algorithm for ($3,p$)-DM}\label{section:dm}

Roughly speaking, the algorithm is based on combining the following lemma from \cite{predivandcol} with the algorithm presented in Section \ref{section:wdm}, as we next describe in more detail.

\begin{lemma}\label{lemma:prevmatch}
If there is a solution to the input, then for any set ${\cal P}\subseteq {\cal S}$ of $p-1$ disjoint tuples, there is a solution to the input whose tuples contain at least $2(p-1)$ elements of tuples in ${\cal P}$.
\end{lemma}
Denote $U=U_1\cup U_2\cup U_3$, and let $<$ be an order on $U$. The algorithm first computes a set ${\cal P}\subseteq {\cal S}$ of $p-1$ disjoint tuples (by using recursion). By Lemma \ref{lemma:prevmatch}, there is $t\in\{1,2,3\}$ such that if there is a solution to the input, then there is a solution to the input whose tuples contain at least $\lceil4(p-1)/3\rceil$ elements in $U\setminus U_t$ that appear in (the tuples of) ${\cal P}$.

For each $t\in\{1,2,3\}$, the algorithm iterates over $U_t$ in an ascending order and over subsets of the set of elements in $U\setminus U_t$ that appear in ${\cal P}$, such that when we reach an element $u\in U_t$ and a subset $P$, we have already computed representative sets of sets of "partial solutions" that include only tuples whose $t^{\mathrm{st}}$ elements are smaller than $u$ and whose set of elements in $U\setminus U_t$ that appear in ${\cal P}$ is a subset of $P$. Then, we try to extend the "partial solutions" by adding tuples whose $t^{\mathrm{st}}$ element is $u$ and computing new representative sets accordingly. The representative sets do not need to hold information on elements in $U\setminus U_t$ that appear in ${\cal P}$ (we store the necessary information on such elements separately). Moreover, the elements in $U_t$ that appear in the "partial solutions" do not appear in any tuple whose $t^{\mathrm{st}}$ element is at least $u$, and any tuple whose $t^{\mathrm{st}}$ element is at least $u$ does not contain elements in $U_t$ that appear in the "partial solutions". We can thus use "better" representative sets, which improves the running time of the algorithm.

We next give the notation used in this section. We then describe the algorithm and give its pseudocode. Finally, we prove its correctness and running time.

\myparagraph{Notation}Let $t\in\{1,2,3\}$ and $P_t\subseteq U\setminus U_t$. Let $u^t_s$ (resp. $u^t_g$) be the smallest (resp. greatest) element in $U_t$. Given $u\in U_t$ and $P\subseteq P_t$, denote ${\cal S}_{t,u,P_t,P}=\{S\in {\cal S}: S$ includes $u$, $P$ is the set of elements in $S$ that appear in $P_t\}$. Given $S\in {\cal S}$, let $\mathrm{set}_{t,P_t}(S)$ be the set of elements in $S$, excluding its $t^{\mathrm{st}}$ element and elements that belong to $P_t$. Given ${\cal S}'\subseteq {\cal S}$, denote $\mathrm{tri}_{t,P_t}({\cal S}')=(\bigcup_{S\in{\cal S}'}\mathrm{set}_{t,P_t}(S),{\cal S'},1)$. Given ${\bf S}\subseteq 2^{\cal S}$, denote $\mathrm{tri}_{t,P_t}({\bf S})=\{\mathrm{tri}_{t,P_t}({\cal S}'): {\cal S}'\in {\bf S}\}$.

Given $u\in U_t$, $1\leq i\leq p$ and $P\subseteq P_t$, define $SOL_{t,u,i,P_t,P} = \{{\cal S}'\subseteq\bigcup_{u'\in U_t\ \mathrm{s.t.} u'\leq u, P'\subseteq P}{\cal S}_{t,u',P_t,P'}: |{\cal S'}|=i,$ the tuples in ${\cal S}'$ are disjoint, $P$ is the set of elements of the tuples in ${\cal S}'$ that appear in $P_t\}$. Note that for all $(X,{\cal S}',W)\in \mathrm{tri}_{t,P}(SOL_{t,u,i,P_t,P})$, we have that $|X|=2i-|P|$.

\myparagraph{The Algorithm}We now describe our algorithm for ($3,p$)-DM, that we call \alg{DM-Alg} (see the pseudocode below). In Step \ref{stepdm:p*}, \alg{DM-Alg} computes a set ${\cal P}\subseteq {\cal S}$ of $p-1$ disjoint tuples. Then, in Step \ref{stepdm:tloop}, it iterates over each $t\in\{1,2,3\}$ and $r\in \{0,\ldots,\lfloor(2p+4)/3\rfloor\}$, where $r$ notes the number of elements in $U\setminus U_t$ that do not appear in ${\cal P}$ and should appear in the currently desired solution. Next consider an iteration corresponding to such $t$ and $r$.

\alg{DM-Alg} introduces a matrix M, where each cell M$[u,i,P]$ will hold a subset of $SOL_{t,u,i,P_t,P}$. It then iterates over $U_t$ in an ascending order and over every subset $P$ of $P_t$ s.t. $2-r\leq|P|\leq 2p-r$. In each iteration, corresponding to such $u$ and $P$, \alg{DM-Alg} computes any cell of the form M$[u,i,P]$ s.t. $1\leq i\leq p$ by using M$[u',i,P]$ and M$[u',i-1,P']$ for all $P'\subseteq P$ (where $u'$ is the element preceding $u$ in $U_t$). In other words, for any $1\leq i\leq p$, \alg{DM-Alg} computes a subset of $SOL_{t,u,i,P_t,P}$ by using subsets of $SOL_{t,u',i,P_t,P}$ and $\bigcup_{P'\subseteq P}SOL_{t,u',i-1,P_t,P'}$. If there is a solution containing exactly $2p-r$ elements from $U\setminus U_t$ that appear in $P_t$, then by using representative sets, \alg{DM-Alg} guarantees that each cell M$[u,i,P]$ will hold "enough" sets from $SOL_{t,u,i,P_t,P}$, such that when the computation of M is finished, $\bigcup_{P\subseteq P_t}\mathrm{M}[u^t_g,p,P]$ will hold some ${\cal S}'\in \bigcup_{P\subseteq P_t}SOL_{t,u^t_g,p,P_t,P}$ (clearly, such a set ${\cal S}'$ is a solution). Moreover, by using representative sets, \alg{DM-Alg} guarantees that each cell M$[u,i,P]$ will not hold "too many" sets from $SOL_{t,u,i,P_t,P}$, since then we will not get an improved running time.

We now describe an iteration of Step \ref{step:dmuploop}, corresponding to some $u$ and $P$, in more detail. By using \alg{R-Alg}, \alg{DM-Alg} first computes a set that $(r-(2-|P|))$-represents tri($SOL_{t,u,1,P_t,P}$) (in Step \ref{stepdm:mu1a}), and assigns its corresponding set of sets of tuples to M$[u,1,P]$ (in Step \ref{stepdm:mu1b}). If $u=u_s$, then $SOL_{t,u,i,P_t,P}$ is empty for all $2\leq i\leq p$, and \alg{DM-Alg} skips the rest of the iteration accordingly (thus M$[u,i,P]$ stays empty, as it is initialized, for all $2\leq i\leq p$). Next assume that $u>u_s$, and consider an iteration of Step \ref{stepdm:loop}, corresponding to some $2\leq i\leq p$. First, in Step \ref{stepdm:step1}, \alg{DM-Alg} computes a set that $(r-(2i-|P|))$-represents tri($SOL_{t,u,i,P_t,P}$) by using the sets in M$[u',i,P]$ and adding tuples in ${\cal S}_{t,u,P_t,P\setminus P'}$ to sets in M$[u',i-1,P']$ for all $P'\subseteq P$. Then, in Step \ref{stepdm:step2}, \alg{DM-Alg} uses \alg{R-Alg} to compute a representative set of the representative set it has just computed in Step \ref{stepdm:step1} in order to reduce its size. Finally, in Step \ref{stepdm:step3}, \alg{DM-Alg} assigns the corresponding set of sets of tuples to M$[u,i,P]$.

\renewcommand{\thealgorithm}{2}
\begin{algorithm}[!ht]
\caption{\alg{DM-Alg}($U_1,\ldots,U_q,{\cal S},p$)}
\begin{algorithmic}[1]\label{alg:dm}
\STATE {\bf if} $p=1$ {\bf then return} some set including exactly one tuple in ${\cal S}$.
\STATE\label{stepdm:p*} ${\cal P}\Leftarrow$ \alg{DM-Alg}($U_1,\ldots,U_q,{\cal S},p-1$).
\FOR{$t=1,2,3$ and $r=0,\ldots,\lfloor(2p+4)/3\rfloor$}\label{stepdm:tloop}
	\STATE let $P_t$ be the set of elements of the tuples in ${\cal P}$, excluding those in $U_t$.
	\STATE let M be a matrix that has a cell $[u,i,P]$ for all $u\in U_t, 1\leq i\leq p$ and $P\subseteq P_t$, which is initialized to $\{\}$.
	\FORALL{$u\in U_t$ {\bf ascending} and $P\subseteq P_t$ s.t. $2-r\leq|P|\leq 2p-r$}	\label{step:dmuploop}	
		\STATE\label{stepdm:mu1a} $\widehat{\cal A}\Leftarrow$\alg{R-Alg}($U,2-|P|,r-(2-|P|),\mathrm{tri}_{t,P}$($\{\{S\}: S\in\bigcup_{u'\in U_t\ \mathrm{s.t.}\ u'\leq u}{\cal S}_{t,u',P_t,P}\}$)).
		\STATE\label{stepdm:mu1b} M$[u,1,P]\Leftarrow\{{\cal S'}: \exists X$ s.t. $(X,{\cal S'},1)\in\widehat{\cal A}\}$.			
		\STATE {\bf if} $u=u^t_s$ {\bf then} skip the iteration. {\bf else} let $u'$ be the element preceding $u$ in~$U_t$.
		\FOR{$i=2,\ldots,\lfloor\frac{|P|+r}{2}\rfloor$}\label{stepdm:loop}
			\STATE\label{stepdm:step1} ${\cal A}\Leftarrow \mathrm{tri}_{t,P}($M$[u',i,P]\cup \{{\cal S}'\cup\{S\}: \exists P'\subseteq P$ s.t. ${\cal S}'\in\mathrm{M}[u',i-1,P'], S\in {\cal S}_{t,u,P_t,P\setminus P'},$ no tuple in ${\cal S}'$ includes an element in $S\})$.
			\STATE\label{stepdm:step2} $\widehat{{\cal A}}\Leftarrow$ \alg{R-Alg}($U,2i-|P|,r-(2i-|P|),{\cal A}$).
			\STATE\label{stepdm:step3} M$[u,i,P]\Leftarrow\{{\cal S'}: \exists X$ s.t. $(X,{\cal S'},1)\in\widehat{\cal A}\}$.			
		\ENDFOR
	\ENDFOR
	\STATE {\bf if} $\bigcup_{P\subseteq P_t}\mathrm{M}[u^t_g,p,P]\neq\emptyset$ {\bf then return} ${\cal S}'\in\bigcup_{P'\subseteq P}\mathrm{M}[u^t_g,p,P]$.
\ENDFOR
\STATE {\bf reject}.
\end{algorithmic}
\end{algorithm}

\myparagraph{Correctness and Running Time}We summarize in the following theorem.

\begin{theorem}
\alg{DM-Alg} solves ($3,p$)-DM in $O(8.04143^p|{\cal S}||U|\log^2|U|)$ deterministic time.
\end{theorem}

\begin{proof}
We prove the theorem by using induction on $p$. For $p=1$, the theorem clearly holds. Next consider some $p\geq 2$ and assume that the theorem holds for all $1\leq p'<p$. By the induction hypothesis, the set ${\cal P}$ computed in Step \ref{stepdm:p*} contains (exactly) $p-1$ disjoint tuples from $\cal S$.

Clearly, for all $1\leq t\leq 3, 0\leq r\leq\lfloor(2p+4)/3\rfloor$ and $P\subseteq P_t$ s.t. $|P|\leq 2p-r$, we have that any ${\cal S}'\in SOL_{t,u^t_g,p,P_t,P}$ is a solution to the input. Now, suppose that there is a solution to the input. By Lemma \ref{lemma:prevmatch}, there is a solution ${\cal S}'$ to the input whose tuples contain at least $2(p-1)$ elements of tuples in ${\cal P}$. Thus, there are $1\leq t\leq 3$, $0\leq r\leq\lfloor(2p+4)/3\rfloor$ and $P\subseteq P_t$ s.t. $|P|=2p-r$, for which $SOL_{t,u^t_g,p,P_t,P}\neq\emptyset$. Thus, the following lemma implies the correctness of the algorithm.

\begin{lemma}
Consider an iteration of Step \ref{stepdm:tloop}, corresponding to some $1\leq t\leq 3$ and $0\leq r\leq\lfloor(2p+4)/3\rfloor$.
For all $u\in U_t$, $1\leq i\leq p$ and $P\subseteq P_t$ s.t. $(|P|\leq 2p-r\wedge i\leq\lfloor\frac{|P|+r}{2}\rfloor)$, $\mathrm{tri}_{t,P}(\mathrm{M}[u,i,P])$ $(r-(2i-|P|))$-represents $\mathrm{tri}_{t,P}$($SOL_{t,u,i,P_t,P}$).
\end{lemma}

\begin{proof}
We prove the lemma by using induction on the order of the computation of M. For all $u\in U_t$ and $P\subseteq P_t$ s.t. $2-r\leq |P|\leq 2p-r$, $SOL_{t,u,1,P_t,P} = \{\{S\}: S\in\bigcup_{u'\in U_t\ \mathrm{s.t.}\ u'\leq u}{\cal S}_{t,u',P_t,P}\}$; and thus, by the definition of \alg{R-Alg} and Steps \ref{stepdm:mu1a} and \ref{stepdm:mu1b}, tri$_{t,P}$(M[$u,1,P$]) $(r-(2-|P|))$-represents tri$_{t,P}$($SOL_{t,u,1,P_t,P}$). For all $P\subseteq P_t$ s.t. $|P|\leq 2p-r$ and $2\leq i\leq \lfloor\frac{|P|+r}{2}\rfloor$, $SOL_{t,u^t_s,i,P_t,P} = \{\}$; and thus, by the initialization of M, tri$_{t,P}$(M[$u^t_s,i,P$]) $(r-(2i-|P|))$-represents tri$_{t,P}$($SOL_{t,u^t_s,i,P_t,P}$).

Next consider an iteration of Step \ref{stepdm:loop} that corresponds to some $u\in U_t\setminus\{u^t_s\}, P\subseteq P_t$ s.t. $|P|\leq 2p-r$ and $2\leq i\leq \lfloor\frac{|P|+r}{2}\rfloor$, and assume that the lemma holds for for the element $u'$ preceding $u$ in $U_t$, all $P'\subseteq P$ and all $1\leq i'\leq\mathrm{min}\{i,\lfloor\frac{|P'|+r}{2}\rfloor\}$. By the definition of \alg{R-Alg}, Observation \ref{obs:transitive} and Steps \ref{stepdm:step2} and \ref{stepdm:step3}, it is enough to prove that $\cal A$ $(r-(2i-|P|))$-represents tri$_{t,P}$($SOL_{t,u,i,P_t,P}$).

By the induction hypothesis and Step \ref{stepdm:step1}, we have that ${\cal A} \subseteq \mathrm{tri}_{t,P}$($SOL_{t,u,i,P_t,P}$). Assume that there are $(X,{\cal S}',1)\in\mathrm{tri}(SOL_{t,u,i,P_t,P})$ and $Y\subseteq U\setminus X$ s.t. $|Y|\leq r-(2i-|P|)$, since otherwise the lemma clearly holds. We have two possible cases as follows.
\begin{enumerate}
\item For all $P'\subseteq P$, ${\cal S}'\cap {\cal S}_{t,u,P_t,P\setminus P'}=\emptyset$. Note that ${\cal S}'\in SOL_{t,u',i,P_t,P}$. Thus, by the induction hypothesis, there is $(X^*,{\cal S}^*,1)\in\mathrm{tri(M}[u',i,P])$ s.t. $X^*\cap Y=\emptyset$; and therefore $(X^*,{\cal S}^*,1)\in{\cal A}$.

\item There is $P'\subseteq P$ s.t. ${\cal S}'\cap {\cal S}_{t,u,P_t,P\setminus P'}=\{S\}$ for some $S$. Note that $|P'|\leq 2p-r$, $i-1\leq \lfloor\frac{|P'|+r}{2}\rfloor$ and ${\cal S}'\setminus\{S\}\in SOL_{t,u',i-1,P_t,P'}$. Thus, by the induction hypothesis, there is $(X^*,{\cal S}^*,1)\in\mathrm{tri}_{t,P'}\mathrm{(M}[u',i-1,P'])$ s.t. $X^*\cap (Y\cup\mathrm{set}_{t,P_t}(S))=\emptyset$. We get that $(X^*\cup\mathrm{set}_{t,P_t}(S),{\cal S}^*\cup\{S\},1)\in{\cal A}$.
\end{enumerate}
We get that there is $(X^*,{\cal S}^*,1)\in{\cal A}$ s.t. $X^*\cap Y=\emptyset$.\qed
\end{proof}
By the induction hypothesis, the definition of \alg{R-Alg} and the pseudocode, the algorithm runs in time

\[O(\sum_{t=1}^3\sum_{r=0}^{\lfloor2p/3\rfloor}\sum_{u\in U_t}\sum_{P\subseteq P_t\ \mathrm{s.t.}\ 2-r\leq|P|\leq 2p-r}\sum_{i=1}^{\lfloor\frac{|P|+r}{2}\rfloor}[{r \choose 2i-|P|}2^{o(r)}|{\cal S}|(\frac{r}{r-(2i-|P|)})^{r-(2i-|P|)}\log^2 |U|]) =\]

\[O(4^p2^{o(p)}|{\cal S}||U|\log^2|U|\mathrm{max}_{t=0}^{\lfloor2p/3\rfloor}{\lfloor2p/3\rfloor \choose t}(\frac{\lfloor2p/3\rfloor}{\lfloor2p/3\rfloor-t})^{\lfloor2p/3\rfloor-t})\]
The maximum is achieved at $t=\alpha\lfloor2p/3\rfloor$, where $\alpha = 1 + \frac{1-\sqrt{1+4e}}{2e}$. Thus, the running time of the algorithm is $O(4^p\cdot 2.850423^{2p/3}|{\cal S}||U|\log^2|U|) = O(8.04143^p\cdot|{\cal S}||U|\log^2|U|)$.\qed
\end{proof}

\mysection{An Algorithm for ($q,p$)-WSP}\label{section:wsp}

Let $<$ be an order on $U$. Roughly speaking, the algorithm is based on combining the following lemma from \cite{predivandcol} with the algorithm presented in Section \ref{section:wdm}, as we next describe in more detail.

\begin{lemma}\label{lemma:wspobs}
Let ${\cal S}'\subseteq{\cal S}$, and denote $S_\mathrm{min} = \{u: \exists S\in{\cal S}'$ s.t. $u$ is the smallest element in $S\}$. Then, any $S\in {\cal S}$ whose smallest element is greater than $\mathrm{max}(S_\mathrm{min})$ does not contain any element from $S_\mathrm{min}$.
\end{lemma}
The algorithm iterates over $U$ in an ascending order, such that when we reach an element $u\in U$, we have already computed representative sets of sets of "partial solutions" that include only sets whose smallest elements are smaller than $u$. Then, we try to extend the "partial solutions" by adding sets whose smallest element is $u$ and computing new representative sets accordingly. By Lemma \ref{lemma:wspobs}, the elements in $U$ that are the smallest elements of sets in the "partial solutions" do not appear in any set whose smallest element is at least $u$. This allows us to use "better" representative sets, which improves the running time of the algorithm. We note that the sets in the "partial solutions" can contain $u$ (and elements greater than $u$); thus the running time of \alg{WDM-Alg} (see Section \ref{section:wdm}) is better than the running time of the algorithm presented in this section.

We next give the notation used in this section. Since the algorithm is similar to \alg{WDM-Alg} (see Section \ref{section:wdm}), we only give its pseudocode. Finally, we prove its correctness and running time.

\myparagraph{Notation}Let $u_s$ (resp. $u_g$) be the smallest (resp. greatest) element in $U$. Given $u\in U$, denote ${\cal S}_u=\{S\in{\cal S}: u$ is the smallest element in~$S\}$. Given a set $S$, let $\mathrm{set}(S)$ be the set of elements in $S$, excluding its smallest element. Given a set of sets ${\cal S}'$, denote $\mathrm{tri}({\cal S}')=(\bigcup_{S\in{\cal S}'}\mathrm{set}(S),{\cal S}',\sum_{S\in{\cal S}'}w(S))$. Given a set of sets of sets ${\bf S}$, denote $\mathrm{tri}({\bf S})=\{\mathrm{tri}({\cal S}'): {\cal S}'\in{\bf S}\}$. Given $S\in{\cal S}$ and $1\leq j\leq q$, let $S_j$ denote the set including the $j$ smallest elements in $S$, and define $w(S_j)=w(S)$.

Given $u\in U$ and $1\leq i\leq p$, define $SOL_{u,i}= \{{\cal S}'\subseteq\bigcup_{u'\in U\ \mathrm{s.t.}\ u'\leq u}{\cal S}_u': |{\cal S}'|=i,$ the sets in ${\cal S}'$ are disjoint$\}$. Note that for all $(X,{\cal S}',W)\in\mathrm{tri}(SOL_{u,i})$, we have that $|X|=(q-1)i$. Given also $S\in {\cal S}_u$ and $1\leq j\leq q$, define $SOL_{u,i,S,j}= \{{\cal S}'\subseteq\{S_j\}\cup(\bigcup_{u'\in U\ \mathrm{s.t.}\ u'<u}{\cal S}_u'): S_j\in {\cal S}', |{\cal S}'|=i,$ the sets in ${\cal S}'$ are disjoint$\}$. Note that for all $(X,{\cal S}',W)\in\mathrm{tri}(SOL_{u,i,S,j})$, we have that $|X|=(q-1)(i-1)+j-1$.

\myparagraph{The Algorithm}The pseudocode of our algorithm for ($q,p$)-WSP, called \alg{WSP-Alg}, is given below.

\renewcommand{\thealgorithm}{3}
\begin{algorithm}[!ht]
\caption{\alg{WSP-Alg}($U,{\cal S},w,p$)}
\begin{algorithmic}[1]\label{alg:wsp}
\STATE let M be a matrix that has a cell $[u,i]$ for all $u\in U$ and $1\leq i\leq p$, which is initialized to $\{\}$.
\FORALL{$u\in U$ {\bf ascending}}
	\STATE $\widehat{\cal A}\Leftarrow$ \alg{R-Alg}($U,q-1,q(p-1),$tri($\{\{S\}: S\in\bigcup_{u'\in U\ \mathrm{s.t.} u'\leq u}{\cal S}_{u'}\}$)).
	\STATE M$[u,1]\Leftarrow\{{\cal S'}: \exists X,W$ s.t. $(X,{\cal S'},W)\in\widehat{\cal A}\}$.	
	\STATE {\bf if} $u=u_s$ {\bf then} skip the iteration. {\bf else} let $u'$ be the element preceding $u$ in $U$.
	\FOR{$i=2,\ldots,p$}
		\STATE ${\cal A}\Leftarrow \mathrm{tri}($M$[u',i]) \cup (\bigcup_{S\in {\cal S}_u}$\alg{WSP-Add}($i,S,\mathrm{M}[u',i-1])$).
		\STATE $\widehat{{\cal A}}\Leftarrow$ \alg{R-Alg}($U,(q-1)i,q(p-i),{\cal A}$).
		\STATE M$[u,i]\Leftarrow\{{\cal S'}: \exists X,W$ s.t. $(X,{\cal S'},W)\in\widehat{\cal A}\}$.
	\ENDFOR
\ENDFOR
\STATE {\bf if} M$[u_g,p]=\emptyset$ {\bf then reject}. {\bf else return} ${\cal S}'\in\mathrm{M}[u_g,p]$ that maximizes $\sum_{S\in{\cal S}'}w(S)$.
\end{algorithmic}
\end{algorithm}

\begin{algorithm}[!ht]
\caption{\alg{WSP-Add}($i,S,{\bf S}$)}
\begin{algorithmic}[1]\label{alg:wspadd}
\STATE $\widehat{\cal B}_1\Leftarrow \{(X,{\cal S}'\cup\{S_1\}, W+w(S)): (X,{\cal S}',W)\in \mathrm{tri}({\bf S}),$ no set in ${\cal S}'$ includes the element in $S_1\}$.
\FOR{$j=2,\ldots,q$}
	\STATE ${\cal B}_j\Leftarrow \{(X\cup\{u_j\},({\cal S}'\setminus\{S_{j-1}\})\cup\{S_j\}, W): (X,{\cal S}',W)\in \widehat{\cal B}_{j-1}, u_j$ is the $j^\mathrm{st}$ smallest element in $S, u_j\notin X\}$.
	\STATE $\widehat{\cal B}_j\Leftarrow$ \alg{R-Alg}($U,(q-1)(i-1)+(j-1),q(p-i)+(q-j),{\cal B}_j$).
\ENDFOR
\STATE {\bf return} $\widehat{\cal B}_q$.
\end{algorithmic}
\end{algorithm}

\myparagraph{Correctness and Running Time}By using the new definitions of $\mathrm{set}()$ and $\mathrm{tri}()$, the next lemma can be proved similarly to Lemma \ref{lemma:wdmadd} (see Appendix \ref{app:proofs}).

\begin{lemma}\label{lemma:wspadd}
Given $2\leq i\leq p$, $S\in S_u$ for some $u\in U$, and ${\bf S}$ s.t. $\mathrm{tri}({\bf S})$ $q(p-(i-1))$-represents $\mathrm{tri}(SOL_{u',i-1})$ where $u'$ is the element preceding $u$ in $U$, \alg{WSP-Add} returns a set that $q(p-i)$-represents $\mathrm{tri}(SOL_{u,i,S,q})$.
\end{lemma}

\begin{theorem}
\alg{WSP-Alg} solves ($q,p$)-WSP in $O((0.56201\cdot2.85043^{q})^p|{\cal S}||U|\log^2|U|)$ deterministic time. In particular, it solves ($3,p$)-WSP in $O^*(12.15493^p)$ deterministic time.
\end{theorem}

\begin{proof}
By using the new definitions of $\mathrm{set}()$ and $\mathrm{tri}()$, the next lemma, which clearly implies the correctness of the algorithm, can be proved similarly to Lemma \ref{lemma:wdm} (see Appendix \ref{app:proofs}).

\begin{lemma}\label{lemma:wsp}
For all $u\in U$ and $1\leq i\leq p$, $\mathrm{tri}(\mathrm{M}[u,i])$ $q(p-i)$-represents $\mathrm{tri}(SOL_{u,i})$.
\end{lemma}
Denote $x= 2^{o(qp)}|{\cal S}||U|$$\log^2$$|U|$. By the definition of \alg{R-Alg} and the pseudocode, the algorithm runs in~time

\[O(\sum_{u\in U}\sum_{i=1}^{p}\sum_{j=1}^q[{qp-i \choose (q-1)(i-1)+j-1}2^{o(qp)}|{\cal S}|(\frac{qp-i}{qp-qi+q-j})^{qp-qi+q-j}\log^2 |U|)] =\]

\[O(x\cdot\mathrm{max}_{i=1}^p\mathrm{max}_{j=1}^q\left\{{qp-i \choose qi-i-q+j}(\frac{qp-i}{qp-qi+q-j})^{qp-qi+q-j}\right\}) =\]

\[O(x\cdot\mathrm{max}_{t=1}^{qp}\left\{{qp-\lceil(t/q)\rceil \choose t-\lceil(t/q)\rceil}(\frac{qp-\lceil(t/q)\rceil}{qp-t})^{qp-t}\right\}) = O(x\cdot\mathrm{max}_{t=1}^{qp}\left\{\frac{(qp-(t/q))^{2qp-t-(t/q)}}{(t-(t/q))^{t-(t/q)}(qp-t)^{2qp-2t}}\right\}) =\]

\[O(x\cdot\mathrm{max}_{0<\alpha<1}\left\{\frac{(qp-\alpha p)^{2qp-\alpha qp-\alpha p}}{(\alpha qp-\alpha p)^{\alpha qp-\alpha p}(qp-\alpha qp)^{2qp-2\alpha qp}}\right\})\]

\[O(x\cdot\mathrm{max}_{0<\alpha<1}\left\{[\frac{(q-\alpha)^{2q-\alpha q-\alpha}}{(\alpha q-\alpha)^{\alpha q-\alpha}(q-\alpha q)^{2q-2\alpha q}}]^p\right\}) =\]

\[O(x\cdot[\mathrm{max}_{0<\alpha<1}\left\{(\frac{\alpha q-\alpha}{q-\alpha})^{\alpha}(\frac{(q-\alpha)^{2-\alpha}}{(\alpha q-\alpha)^{\alpha}(q-\alpha q)^{2-2\alpha}})^q\right\}]^p) = (*)\]
When $q=3$, the maximum of (*) is achieved at $\alpha\cong 0.58226$. Thus, \alg{WSP-Alg} solves ($3,p$)-WSP in $O^*(12.15493^p)$ deterministic time. Now, note that

\[(*) = O(x\cdot[\mathrm{max}_{0<\alpha<1}\left\{(\frac{\alpha}{e^{1-\alpha}})^\alpha(\frac{1}{\alpha^\alpha(1-\alpha)^{2-2\alpha}})^q\right\}]^p)\]
As we increase $q$, the $\alpha$ for which we get the maximum decreases, staying greater than $\alpha^* = 1 + \frac{1-\sqrt{1+4e}}{2e}$ (since this $\alpha^*$ maximizes $(\frac{1}{\alpha^{\alpha}(1-\alpha)^{2-2\alpha}})^q$). When $q=1,500$, the maximum of (*) is achieved at $\alpha' < 0.550148$, and thus when $q\geq 1,500$, we get that \alg{WSP-Alg} runs in time $O(x\cdot(\frac{\alpha'}{e^{1-\alpha'}})^{\alpha'}(\frac{1}{{\alpha^*}^{\alpha^*}(1-\alpha^*)^{2-2\alpha^*}})^q) = O(x\cdot (0.56201\cdot2.85043^{q})^p)$. Since this expression bounds (*) for smaller values for $q$, we get the desired running time.\qed
\end{proof}

\mysection{Kernels for $(q,p)$-WDM and $(q,p)$-WSP}\label{section:kernel}

We first give the notation used in this section. Then we present our kernel for $(q,p)$-WDM, followed by our kernel for $(q,p)$-WSP. Finally, by using these kernels, we improve the running times (though not the $O^*$ running times) of the algorithms presented in the previous three sections. In this section, given an input to $(q,p)$-WDM or $(q,p)$-WSP, assume that any element in the universe(s) appears in some tuple$\setminus$set in ${\cal S}$, since otherwise we can delete it.

\myparagraph{Notation}Given a tuple or a set $S$, let $\mathrm{set}(S)$ be the set of elements in $S$. Given a set of tuples or sets ${\cal S}'$, denote $\mathrm{tri}({\cal S}') = \{(\mathrm{set}(S),S,w(S)): S\in {\cal S}\}$.

\myparagraph{A Kernel for $(q,p)$-WDM}We now present a kernelization algorithm, that we call \alg{WDM-Ker}, for $(q,p)$-WDM (see the pseudocode below).

\renewcommand{\thealgorithm}{4}
\begin{algorithm}[!ht]
\caption{\alg{WDM-Ker}($U_1,\ldots,U_q,{\cal S},w,p$)}
\begin{algorithmic}[1]\label{ker:wdm}
\STATE\label{step:ker1} {\bf if} $|{\cal S}|\leq e^q(p-1)^q$ {\bf then return} ($U_1,\ldots,U_q,{\cal S},w,p$).
\STATE\label{step:ker2} $\widehat{\cal A}\Leftarrow$ \alg{K-Alg}($\bigcup_{i=1}^qU,q,q(p-1),$tri(${\cal S}$)).
\STATE\label{step:ker3} {\bf for} $i=1,\ldots,q$ {\bf do} $U^*_i\Leftarrow\{u\in U_i: \exists (X,{\cal S}',W)\in\widehat{\cal A}$ s.t. $u\in X\}$.
\STATE\label{step:ker4} {\bf return} ($U^*_1,\ldots,U^*_q,\{S: \exists X,W\ \mathrm{s.t.}\ (X,S,W)\in\widehat{A}\},w,p$).
\end{algorithmic}
\end{algorithm}

\begin{theorem}\label{theorem:ker1}
Given an input ($U_1,\ldots,U_q,{\cal S},w,p$) for $(q,p)$-WDM, \alg{WDM-Ker} returns an input ($U^*_1,\ldots,$\\$U^*_q,{\cal S}^*,w,p$) for $(q,p)$-WDM, s.t. $\sum_{i=1}^q|U^*_i|\leq q|{\cal S}^*|$, $|{\cal S}^*| = O(e^q(p-1)^q)$, and a set ${\cal S}'$ solves ($U_1,\ldots,U_q,$\\${\cal S},w,p$) iff there is a solution ${\cal S}''$ to ($U^*_1,\ldots,U^*_q,{\cal S}^*,w,p$) s.t. $\sum_{S\in{\cal S}'}w(S)=\sum_{S\in{\cal S}''}w(S)$. \alg{WDM-Ker} runs in time $O([\mathrm{min}(|{\cal S}|,e^q(p-1)^q)]^{\tilde{w}-1}|{\cal S}|q^2\log|\bigcup_{i=1}^qU|)$.
\end{theorem}

\begin{proof}
If $|{\cal S}|\leq e^q(p-1)^q$, then by Step \ref{step:ker1}, the algorithm is clearly correct and runs in the desired time; thus next assume that $|{\cal S}| > e^q(p-1)^q$.

By the definition of \alg{K-Alg} and Steps \ref{step:ker2}--\ref{step:ker4}, we get that $\sum_{i=1}^q|U^*_i|\leq q|{\cal S}^*|$ and $|{\cal S}^*|\leq{qp \choose q} = O(\frac{p^{qp}}{(p-1)^{qp-q}}) = O(e^q(p-1)^q)$. Moreover, we get that the algorithm runs in time bounded by

\[O(|{\cal S}|{qp \choose q}^{\tilde{w}-1}\log(q!|\bigcup_{i=1}^qU|^{q^2})) = O(|{\cal S}|(e^q(p-1)^q)^{\tilde{w}-1}q^2\log|\bigcup_{i=1}^qU|).\]
By the definition of \alg{K-Alg} and Steps \ref{step:ker2}--\ref{step:ker4}, we get that ($\forall i\in\{1,\ldots,q\}$: $U^*_i\subseteq U_i$) and ${\cal S}^*\subseteq {\cal S}$. Thus, if ($U_1,\ldots,U_q,{\cal S},w,p$) does not have a solution, then ($U^*_1,\ldots,U^*_q,{\cal S}^*,w,p$) does not have a solution, and if a set ${\cal S}'$ is a solution to ($U_1,\ldots,U_q,{\cal S},w,p$), then ($U^*_1,\ldots,U^*_q,{\cal S}^*,w,p$) does not have a solution ${\cal S}''$ s.t. $\sum_{S\in{\cal S}'}w(S)<\sum_{S\in{\cal S}''}w(S)$.

It is now enough to prove that given a solution ${\cal S}'$ to ($U_1,\ldots,U_q,{\cal S},w,p$), there is a set of disjoint tuples ${\cal S}''\subseteq {\cal S}^*$ s.t. $\sum_{S\in{\cal S}'}w(S)\leq\sum_{S\in{\cal S}''}w(S)$. Consider the following lemma.

\begin{lemma}\label{lemma:wdmker}
Let ${\cal S}'=\{S'_1,\ldots,S'_p\}$ be a solution to ($U_1,\ldots,U_q,{\cal S},w,p$). For all $i\in\{0,\ldots,p\}$, there is a set of disjoint tuples ${\cal S}^*_i=\{S^*_1,\ldots,S^*_i\}\subseteq {\cal S}^*$ s.t. $\sum_{j=1}^i w(S'_j)\leq \sum_{j=1}^i w(S^*_j)$, whose tuples are disjoint from those in $\{S'_{i+1},\ldots,S'_p\}$.
\end{lemma}

\begin{proof}
We prove the lemma by using induction on $i$. The claim clearly holds for $i=0$, since then we can choose ${\cal S}^*_0=\{\}$. Next consider some $i\in\{1,\ldots,p\}$ and assume that the claim holds for $i-1$. By the induction hypothesis, there is a set of disjoint tuples ${\cal S}^*_{i-1}=\{S^*_1,\ldots,S^*_{i-1}\}\subseteq {\cal S}^*$ s.t. $\sum_{j=1}^{i-1} w(S'_j)\leq \sum_{j=1}^{i-1} w(S^*_j)$, whose tuples are disjoint from those in $\{S'_i,\ldots,S'_p\}$. By the definition of \alg{K-Alg} and Steps \ref{step:ker2} and \ref{step:ker4}, there is a tuple $S^*_i\in{\cal S}^*$ s.t. $w(S'_i)\leq w(S^*_i)$, which is disjoint from the tuples in $\{S^*_1,\ldots,S^*_{i-1},S'_{i+1},\ldots,S'_p\}$. Thus, by defining ${\cal S}^*_i=\{S^*_1,\ldots,S^*_i\}\subseteq {\cal S}^*$, we conclude the lemma.\qed
\end{proof}
Lemma \ref{lemma:wdmker}, implying the existence of the required set, concludes the theorem.\qed
\end{proof}

\myparagraph{A Kernel for $(q,p)$-WSP} By trivial modifications of \alg{WDM-Ker} (see Appendix~\ref{app:kerwsp}), we get a kernelization algorithm, that we call \alg{WSP-Ker}, which satisfies the following result.

\begin{theorem}\label{theorem:ker2}
Given an input ($U,{\cal S},w,p$) for $(q,p)$-WSP, \alg{WSP-Ker} returns an input ($U^*,{\cal S}^*,w,p$) for $(q,p)$-WSP, s.t. $|U^*|\leq q|{\cal S}^*|$, $|{\cal S}^*| = O(e^q(p-1)^q)$, and a set ${\cal S}'$ solves ($U,{\cal S},w,p$) iff there is a solution ${\cal S}''$ to ($U^*,{\cal S}^*,w,p$) s.t. $\sum_{S\in{\cal S}'}w(S)=\sum_{S\in{\cal S}''}w(S)$. \alg{WSP-Ker} runs in time $O([\mathrm{min}(|{\cal S}|,e^q(p-1)^q)]^{\tilde{w}-1}|{\cal S}|q^2\log|U|)$.
\end{theorem}

\myparagraph{Improving the Running Times of \alg{WDM-Alg}, \alg{DM-Alg} and \alg{WSP-Alg}}Since $e^qq(p-1)^q=O(2^{O(q\log p)})$\\$=O(2^{o(qp)})$, Theorems \ref{theorem:wdm}--\ref{theorem:ker2} imply the following results.

\begin{itemize}
\item ($q,p$)-WDM can be solved in $O(2^{o(qp)}|{\cal S}|\log|U| + 2.85043^{(q-1)p})$ deterministic time. In particular, ($3,p$)-WDM can be solved in $O(2^{o(qp)}|{\cal S}|\log|U| + 8.12492^p)$ deterministic time.
\item ($3,p$)-DM can be solved in $O(2^{o(qp)}|{\cal S}|\log|U| + 8.04143^p)$ deterministic time.
\item ($q,p$)-WSP can be solved in $O(2^{o(qp)}|{\cal S}|\log|U| + (0.56201\cdot2.85043^{q})^p)$ deterministic time. In particular, ($3,p$)-WSP can be solved in $O(2^{o(qp)}|{\cal S}|\log|U| + 12.15493^p)$ deterministic time.
\end{itemize}

\bibliographystyle{splncs03}
\bibliography{references}

\appendix

\section{Some Proofs}\label{app:proofs}

\subsection{Proof of Lemma \ref{lemma:wspadd}}
By using induction on $j$, we prove that for all $1\leq j\leq q$, $\widehat{\cal B}_j$ $(q(p-i)+(q-j))$-represents tri($SOL_{u,i,S,j}$). By Step \ref{stepwdmadd:b1}, since tri$({\bf S})$ $q(p-(i-1))$-represents tri($SOL_{u',i-1}$), we have that $\widehat{\cal B}_1$ $(q(p-i)+(q-1))$-represents tri($SOL_{u,i,S,1}$).

Next consider some $2\leq j\leq q$, and assume that the lemma holds for all $1\leq j'<j$. By the definition of \alg{R-Alg}, Observation \ref{obs:transitive} and Step \ref{stepwdmadd:hatbj}, it is enough to prove that ${\cal B}_j$ $(q(p-i)+(q-j))$-represents tri$(SOL_{u,i-1,S,j})$.

By the induction hypothesis and Step \ref{stepwdmadd:bj}, we get that ${\cal B}_j\subseteq \mathrm{tri}(SOL_{u,i,S,j})$. Assume that there are $(X,{\cal S}',W)\in\mathrm{tri}(SOL_{u,i,S,j})$ and $Y\subseteq U\setminus X$ s.t. $|Y|\leq (q(p-i)+(q-j))$, since otherwise the claim clearly holds. Let $u_j$ be the $j^\mathrm{st}$ smallest element in $S$. Note that $(X\setminus\{u_j\},({\cal S}'\setminus\{S_j\})\cup\{S_{j-1}\},W)\in \mathrm{tri}(SOL_{u,i,S,j-1})$. Thus, by the induction hypothesis, there is $(X^*,{\cal S}^*,W^*)\in\widehat{\cal B}_{j-1}$ s.t. $X^*\cap (Y\cup\{u_j\})=\emptyset$ and $W^*\geq W$. We get that $(X^*\cup\{u_j\},({\cal S}^*\setminus\{S_{j-1}\})\cup\{S_j\},W^*)\in{\cal B}_j$. Since $(X^*\cup\{u_j\})\cap Y=\emptyset$ and $W^*\geq W$, we get that the claim holds.\qed

\subsection{Proof of Lemma \ref{lemma:wsp}}
We prove the lemma by using induction on the order of the computation of M. For all $u\in U$, $SOL_{u,1} = \{\{S\}: S\in\bigcup_{u'\in U\ \mathrm{s.t.}\ u'\leq u}{\cal S}_{u'}\}$; and thus, by the definition of \alg{R-Alg} and Steps \ref{stepwdm:mu1a} and \ref{stepwdm:mu1b}, tri(M$[u,1]$) $q(p-1)$-represents tri$(SOL_{u,1})$. For all $2\leq i\leq p$, $SOL_{u_s,i} =\{\}$; and thus, by the initialization of M, tri(M$[u_s,i]$) $q(p-i)$-represents tri$(SOL_{u_s,i})$.

Next consider an iteration of Step \ref{stepwdm:loop} that corresponds to some $u\in U\setminus\{u_s\}$ and $2\leq i\leq p$, and assume that the lemma holds for the element $u'$ preceding $u$ in $U$ and all $1\leq i'\leq i$. By the definition of \alg{R-Alg}, Observation \ref{obs:transitive} and Steps \ref{stepwdm:step2} and \ref{stepwdm:step3}, it is enough to prove that $\cal A$ $q(p-i)$-represents tri$(SOL_{u,i})$.

By the induction hypothesis, Step \ref{stepwdm:step1} and Lemma \ref{lemma:wspadd}, we have that ${\cal A}\subseteq \mathrm{tri}(SOL_{u,i})$. Assume that there are $(X,{\cal S}',W)\in\mathrm{tri}(SOL_{u,i})$ and $Y\subseteq U\setminus X$ s.t. $|Y|\leq q(p-i)$, since otherwise the lemma clearly holds. We have two possible cases as follows.
\begin{enumerate}
\item ${\cal S}'\cap {\cal S}_u=\emptyset$. Note that ${\cal S}'\in SOL_{u',i}$. Thus, by the induction hypothesis, there is $(X^*,{\cal S}^*,W^*)\in\mathrm{tri(M}[u',i])$ s.t. $X^*\cap Y=\emptyset$ and $W^*\geq W$; and therefore $(X^*,{\cal S}^*,W^*)\in{\cal A}$.
\item ${\cal S}'\cap {\cal S}_u=\{S\}$ for some $S$. Note that ${\cal S}'\in SOL_{u,i,S,q}$. Thus, by the induction hypothesis and Lemma \ref{lemma:wspadd}, \alg{WSP-Add}$(i,S,\mathrm{M}[u',i-1])$ returns a set that includes a triple $(X^*,{\cal S}^*,W^*)$ s.t. $X^*\cap Y=\emptyset$ and $W^*\geq W$; and therefore $(X^*,{\cal S}^*,W^*)\in{\cal A}$.
\end{enumerate}
We get that there is $(X^*,{\cal S}^*,W^*)\in{\cal A}$ s.t. $X^*\cap Y=\emptyset$ and $W^*\geq W$.\qed

\section{A Kernel for $(q,p)$-WSP}\label{app:kerwsp}

We now present a kernelization algorithm, that we call \alg{WSP-Ker}, for $(q,p)$-WSP (see the pseudocode below).

\renewcommand{\thealgorithm}{4}
\begin{algorithm}[!ht]
\caption{\alg{WSP-Ker}($U,{\cal S},w,p$)}
\begin{algorithmic}[1]\label{ker:wsp}
\STATE {\bf if} $|{\cal S}|\leq e^q(p-1)^q$ {\bf then return} ($U,{\cal S},w,p$).
\STATE $\widehat{\cal A}\Leftarrow$ \alg{K-Alg}($U,q,q(p-1),$tri(${\cal S}$)).
\STATE $U^*\Leftarrow\{u\in U: \exists (X,{\cal S}',W)\in\widehat{\cal A}$ s.t. $u\in X\}$.
\STATE {\bf return} ($U^*,\{S: \exists X,W\ \mathrm{s.t.}\ (X,S,W)\in\widehat{A}\},w,p$).
\end{algorithmic}
\end{algorithm}

\begin{proof}[Theorem \ref{theorem:ker2}]
If $|{\cal S}|\leq e^q(p-1)^q$, then by Step \ref{step:ker1}, the algorithm is clearly correct and runs in the desired time; thus next assume that $|{\cal S}| > e^q(p-1)^q$.

By the definition of \alg{K-Alg} and Steps \ref{step:ker2}--\ref{step:ker4}, we get that $\sum|U^*|\leq q|{\cal S}^*|$ and $|{\cal S}^*|\leq{qp \choose q} = O(\frac{p^{qp}}{(p-1)^{qp-q}}) = O(e^q(p-1)^q)$. Moreover, we get that the algorithm runs in time bounded by

\[O(|{\cal S}|{qp \choose q}^{\tilde{w}-1}\log(q!|U|^{q^2})) = O(|{\cal S}|(e^q(p-1)^q)^{\tilde{w}-1}q^2\log|U|).\]
By the definition of \alg{K-Alg} and Steps \ref{step:ker2}--\ref{step:ker4}, we get that $U^*\subseteq U$ and ${\cal S}^*\subseteq {\cal S}$. Thus, if ($U,{\cal S},w,p$) does not have a solution, then ($U^*,{\cal S}^*,w,p$) does not have a solution, and if a set ${\cal S}'$ is a solution to ($U,{\cal S},w,p$), then ($U^*,{\cal S}^*,w,p$) does not have a solution ${\cal S}''$ s.t. $\sum_{S\in{\cal S}'}w(S)<\sum_{S\in{\cal S}''}w(S)$.

It is now enough to prove that given a solution ${\cal S}'$ to ($U,{\cal S},w,p$), there is a set of disjoint sets ${\cal S}''\subseteq {\cal S}^*$ s.t. $\sum_{S\in{\cal S}'}w(S)\leq\sum_{S\in{\cal S}''}w(S)$. Consider the following lemma.

\begin{lemma}\label{lemma:wspker}
Let ${\cal S}'=\{S'_1,\ldots,S'_p\}$ be a solution to ($U,{\cal S},w,p$). For all $i\in\{0,\ldots,p\}$, there is a set of disjoint sets ${\cal S}^*_i=\{S^*_1,\ldots,S^*_i\}\subseteq {\cal S}^*$ s.t. $\sum_{j=1}^{i} w(S'_j)\leq \sum_{j=1}^{i} w(S^*_j)$, whose sets are disjoint from those in $\{S'_{i+1},\ldots,S'_p\}$.
\end{lemma}

\begin{proof}
We prove the lemma by using induction on $i$. The claim clearly holds for $i=0$, since then we can choose ${\cal S}^*_0=\{\}$. Next consider some $i\in\{1,\ldots,p\}$ and assume that the claim holds for $i-1$. By the induction hypothesis, there is a set of disjoint sets ${\cal S}^*_{i-1}=\{S^*_1,\ldots,S^*_{i-1}\}\subseteq {\cal S}^*$ s.t. $\sum_{j=1}^{i-1} w(S'_j)\leq \sum_{j=1}^{i-1} w(S^*_j)$, whose sets are disjoint from those in $\{S'_i,\ldots,S'_p\}$. By the definition of \alg{K-Alg} and Steps \ref{step:ker2} and \ref{step:ker4}, there is a set $S^*_i\in{\cal S}^*$ s.t. $w(S'_i)\leq w(S^*_i)$, which is disjoint from the sets in $\{S^*_1,\ldots,S^*_{i-1},S'_{i+1},\ldots,S'_p\}$. Thus, by defining ${\cal S}^*_i=\{S^*_1,\ldots,S^*_i\}\subseteq {\cal S}^*$, we conclude the lemma.\qed
\end{proof}
Lemma \ref{lemma:wspker}, implying the existence of the required set, concludes the theorem.\qed
\end{proof}

\end{document}